\documentclass[10pt, conference, letterpaper]{IEEEtran}
\usepackage{amsmath}
\usepackage{amsthm}
\usepackage{amssymb}
\usepackage{verbatim}
\usepackage{graphicx}

\usepackage{tikz}
\usepackage{pgfplots}
\usetikzlibrary {positioning}
\usetikzlibrary{shapes,snakes}
\usetikzlibrary{arrows,calc}
\usepackage{relsize}
\usetikzlibrary{patterns}
 \usepackage{algorithm}
\usepackage{algpseudocode}

\theoremstyle{plain}
\newtheorem{lemma}{Lemma}

\newtheorem{theorem}{Theorem}
\newtheorem{assumption}{Assumption}
\newtheorem{Observation}{Observation}

\theoremstyle{definition}
\newtheorem{example}{Example}
\floatname{algorithm}{Algorithm}

\allowdisplaybreaks[4]

\usepackage{amsfonts}
\usepackage{times}
\usepackage{latexsym}
\usepackage{amssymb}
\usepackage{amsmath}
\usepackage{cite}
\usepackage{verbatim}
\def\bb0{{\mathbb{0}}}

\def\bb{{\mathbf{b}}}

\def\b0{{\mathbf{0}}}
\def\opt{\mathsf{OPT}}
\def\off{\mathsf{OFF}}

\def\b1{{\mathbf{1}}}

\def\bbE{{\mathbb{E}}}

\def\bbP{{\mathbb{P}}}

\def\cA{\mathcal{A}}

\def\cI{\mathcal{I}}

\def\sfb{{\mathsf{b}}}

\def\sf0{{\mathsf{0}}}

\def\nn{\nonumber}

\begin{document}

\newlength{\figurewidth}\setlength{\figurewidth}{0.6\columnwidth}


\addtolength{\topmargin}{-0.5\baselineskip}
\addtolength{\textheight}{\baselineskip}

\title{\fontsize{23}{23}\selectfont Online Knapsack Problem under Expected Capacity Constraint}

\newcounter{one}
\setcounter{one}{1}
\newcounter{two}
\setcounter{two}{2}

\addtolength{\floatsep}{-\baselineskip}
\addtolength{\dblfloatsep}{-\baselineskip}
\addtolength{\textfloatsep}{-\baselineskip}
\addtolength{\dbltextfloatsep}{-\baselineskip}
\addtolength{\abovedisplayskip}{-1ex}
\addtolength{\belowdisplayskip}{-1ex}
\addtolength{\abovedisplayshortskip}{-1ex}
\addtolength{\belowdisplayshortskip}{-1ex}

%
\author{\parbox{3 in}{ \centering Rahul Vaze\\
        School of Technology and Computer Science \\
        Tata Institute of Fundamental Research\\
      	Mumbai, India\\
        {\tt\small vaze@tcs.tifr.res.in}} 
}

\maketitle
\begin{abstract}
Online knapsack problem is considered, where items arrive in a sequential fashion that have two attributes; value and weight. Each arriving item has to be accepted or rejected on its arrival irrevocably. 
The objective is to maximize the sum of the value of the accepted items such that the sum of their weights is below a  budget/capacity.  Conventionally a hard budget/capacity constraint is considered, for which variety of results are available. In modern applications, e.g., in wireless networks, data centres, cloud computing, etc., enforcing the capacity constraint in expectation is sufficient. With this motivation, we consider the knapsack problem with an expected capacity constraint. For the special case of knapsack problem, called the secretary problem, where the weight of each item is unity, we propose an algorithm whose probability of selecting any one of the optimal items is equal to $1-1/e$ and provide a matching lower bound. For the general knapsack problem, we propose an algorithm whose competitive ratio is shown to be $1/4e$ that is significantly better than the best known competitive ratio of $1/10e$ for the knapsack problem with the hard capacity constraint.

\end{abstract}
\section{Introduction} \label{sec:intro}
Knapsack problem \cite{vazirani2001approximation} is a versatile combinatorial object that models a large variety of resource allocation paradigms. The knapsack has a given capacity, and the objective is to choose a set of items with the largest sum of their values such that the sum of their weights/sizes is less than the knapsack capacity. 
There are several important examples of real-life applications of knapsack problem, such as 
allocation of an advertising budget to the promotions of individual
products, allocation of preparation of final exams in different
subjects given limited time, job scheduling in clouds with overall machine time constraint, sensor networks with energy constraints etc. 
Applications of the knapsack problem also include questions in auction design, such as to choose agents with private values and publicly known weights that fit into a knapsack \cite{aggarwal2006knapsack}. 

The knapack problem is is known to be NP-hard even in the offline setting, where an algorithm can select items by considering all items together. It is, however, possible in the offline setting to approximate the optimal solution within a factor of  $1+\epsilon$ for any $\epsilon > 0$ in polynomial time \cite{vazirani2001approximation}. 

The online version of the knapsack problem models the question of resource allocation under the future uncertainties, where items arrive in a sequential fashion and any algorithm has to accept or reject items irrevocably without having access to future arrivals. The online scenario is relevant for applications, such as in cloud servers, where jobs have to accepted/rejected without the knowledge of profitability of future jobs, or to hire a particular candidate not knowing whether a stronger candidate might apply at a later stage, generalized adwords \cite{buchbinder2009design}, load balancing \cite{goel1999stochastic}, cognitive radio, admission control \cite{shi2014online, zhang2014dynamic, zheng2013coordinated, cello10generalized, zhang2009resource}. etc. 
  The performance of any online algorithm is typically quantified using the metric of {\it competitive ratio}, that measures the ratio of the profit of the online algorithm and the optimal offline algorithm (that has access to non-causal information). The online version of the knapsack problem has also received considerable attention in the literature \cite{marchetti1995stochastic, han2015randomized, babaioff2007knapsack, iwama2010online}, with the best known competitive ratio of $1/10e$ in \cite{babaioff2007knapsack}.

In this paper, we consider an important variation of the online knapsack problem, where we enforce the capacity constraint in expectation, that is of both practical and theoretical interest.
Classically, for the knapack problem (both in offline and online cases), a hard capacity constraint is enforced for selecting the items, i.e., the sum of the weight of all the selected items is below a fixed capacity. In modern applications, there are many scenarios, such as cloud computing, where it is sufficient to enforce the capacity constraint in expectation, i.e., for specific instances of input the algorithm might decide to use a larger capacity, but in expectation it satisfies the given constraint. Our expected capacity constraint generalizes the overdraft approach of \cite{BoSrikant} for similar capacitated problems, where performance within $O(\epsilon)$ of the optimal revenue can be achieved when resources of the order of $O(1/\epsilon)$ over the specified budget/capacity constraint are allowed to be used.

One specific application (among other online knapsack applications mentioned above) that motivates the study of the knapsack problem under the expected capacity constraint is job scheduling in clouds, where a large number of jobs are submitted  with heterogenous resource requirements, and it is reasonable to expect the cloud to execute these jobs by using larger memory and resources for some instances of input while maintaining the resource constraint on average so as to maximize its utility function. Similar case can be made for other applications of the knapsack problem such as generalized adwords, load balancing, and sensor network, where it is easy to envisage an expected resource capacity constraint.


An important special case of the online knapsack problem is the secretary problem \cite{ferguson1989solved}, where secretaries are interviewed sequentially, and as soon as one secretary is hired, the process terminates, and no more secretaries are interviewed.
The secretary problem is equivalent to an online knapsack problem, where the weight of each item is $1$ and the hard capacity constraint is also $1$. Thus, in the secretary problem, the objective is to select only one item with the largest value in an online fashion. The secretary problem can be cast as an Markov decision process and has attracted attention from different research communities because of its universality. One limitation of the secretary problem is however that if the input, order of the arrival of items, is controlled by an adversary, then the performance of an optimal algorithm is arbitrarily bad. To keep the problem non-degenerate, a universal assumption is made about the input arrival sequence to be selected via an uniformly random permutation over the set of items, which is also called the secretarial model of input. 

Under the secretary model of input, the classical secretary problem has been solved via multiple approaches as reviewed in \cite{freeman1983secretary}, and the optimal probability (competitive ratio) of selecting the best item is known to be $1/e$. 
Over the years, multiple variants of secretary problems have been studied, that have been well documented in survey \cite{ freeman1983secretary}. Some important variations of the secretary problem include multiple choice \cite{preater1994multiple, kleinberg2005multiple}, infinitely many items \cite{gianini1976infinite}, unknown number of items \cite{freeman1983secretary,ferguson1989solved}, maximizing the expected value \cite{freeman1983secretary,ferguson1989solved}, matroid constraint \cite{babaioff2007matroids}, etc. 
A simple extension of the secretary problem, called the $k$-secretary, is when the objective is to select $k$ items with the largest sum of their values, when all the weights are unity. Similar to the $k=1$-secretary problem, the optimal algorithm has competitive ratio $1/e$ \cite{babaioff2007knapsack} for $k>1$ as well.

To the best of our knowledge, however, the expected capacity constraint, which in the context of the $k$-secretary problem implies that an algorithm can select at most $k$ items in expectation, where the randomness is over the uniformly random input sequence, has not been studied and this new direction is rather novel.  
The fundamental difference between the hard capacity constraint and the expected capacity constraint in the context of secretary or online knapsack problem is that with the expected capacity constraint, an algorithm can scan all the items, and need not terminate as soon as the sum of the weight of selected items is as much as the fixed capacity. For example, given that the values of items in order or their arrival be $\{1,5, 3, 9\}$. Then with a hard capacity constraint of $1$, if the algorithm decides to choose the second item with value $5$, then the algorithm terminates. With an expected capacity constraint of $1$, however, even if an algorithm selects item $2$, it can still consider the two items arriving thereafter and select item $4$ with value $9$. Note that an algorithm is not allowed to remove already selected items, e.g. in this case item $2$. Thus, with expected capacity constraint, the algorithm will have to appropriately modulate the probability of selecting item $2$ and subsequently item $4$. Thus, the expected capacity constraint allows the algorithm a significantly larger flexibility.

The general online knapsack problem has been studied widely \cite{marchetti1995stochastic, han2015randomized, babaioff2007knapsack} under the hard capacity constraint, with the best known competitive ratio of $1/10 e$ reported in \cite{babaioff2007knapsack} for a randomized algorithm under the secretarial input. Under a large market assumption, that requires that the value of any item is `small' compared to the value of the optimal solution, an online algorithm for the knapsack problem with competitive ratio of $1/2e$ is proposed in \cite{VazeInfocom2017}. The stochastic version of the knapsack problem has been studied in \cite{marchetti1995stochastic}, while restricting the ratio of the value and the weight of any items to lie within $[L,H]$, $1/\log\left(\frac{H}{L}\right)$-competitive algorithms have been proposed in \cite{buchbinder2005online, buchbinder2006improved, zhou2008budget}.

Designing online algorithms for knapsack problem in comparison to the secretary problem are significantly more challenging as evident in \cite{marchetti1995stochastic, han2015randomized, babaioff2007knapsack}, primarily because there is no `simple' offline algorithm that can approximate the optimal solution. Indeed, it is possible to approximate the optimal solution of the knapsack problem in the offline setting within a factor of  $1+\epsilon$ for any $\epsilon > 0$ in polynomial time, however, that algorithm is not amenable to be made {\it online}.

As discussed before, for the case of hard capacity constraint, it is easy to see that if the input (values and weights of items) is chosen adversarially, no deterministic online algorithm can have bounded competitive ratio, and no randomized algorithm can have competitive ratio better than $1/n$ for the secretary problem, where $n$ is the total number of items. We show in this paper that even under the expected capacity constraint no randomized algorithm can have competitive ratio better than $1/n$ for the secretary problem under the adversarial input.  
Thus, following the long line of work on secretary and online knapsack problem \cite{babaioff2007knapsack, KorulaPal}, we consider a secretarial input model even when considering the expected capacity constraint, where the order of arrival of items is uniformly random, but their values and weights are allowed to be arbitrary. 

Our contributions are as follows:
\begin{itemize} 
\item For the secretary problem, under the expected capacity constraint of $1$, we propose an algorithm whose competitive ratio is $1-1/e$. To complement the result, we also show that no online algorithm can achieve competitive ratio better than $1-1/e$ under the expected capacity constraint. Compared to the hard capacity constraint of $1$, where the optimal competitive ratio is $1/e$, there is a two-fold improvement in the competitive ratio with the expected capacity constraint.

\item For the $k$-secretary problem, where the objective is to select the $k$ best items and all items  have weight $1$, under the expected capacity constraint of $k$, a simple modification of the algorithm proposed for the $k=1$-secretary problem is shown to achieve a competitive ratio is $1-1/e$, which is also the best possible. 

\item We propose a $1/4e$ competitive algorithm for the online knapsack problem under the expected capacity constraint, which significantly improves the performance of best known algorithm that has competitive ratio of $1/10e$ under the hard capacity constraint \cite{babaioff2007knapsack}. The main idea of the proposed algorithm is to first consider a `simple' offline algorithm that is allowed to use extra capacity $C$, where $1<C\le 2$, that can be shown to provide a $\frac{1}{C-1}$- approximation to the optimal solution of the offline knapsack problem with hard capacity constraint $1$. The `simple' offline algorithm also provides a threshold for selection of items in terms of the ratio of their weight and the value.

Using this threshold, we then make the `simple' offline algorithm, online, using the ideas of sample and price class of algorithms, where the algorithm only observes (but does not select any) an initial set of items and builds a threshold, which is then used to select the forthcoming items. This online algorithm is shown to be $1/2e$ competitive with respect to the simple offline algorithm, which itself is $\frac{1}{C-1}$ approximate with respect to the optimal offline algorithm for the knapsack problem with hard capacity of $1$. The online algorithm that uses capacity $C$ is then used with probability $1/C$ to ensure the expected capacity constraint of $1$ and results in overall competitive ratio of $\frac{1}{2e C(C-1)}$, which is $=\frac{1}{4e}$ for $C=2$.
This algorithm's performance comes close to $1/2e$-competitive algorithm of \cite{VazeInfocom2017}, which however is valid only under the large market assumption. 

\end{itemize}

\section{Online Knapsack Problem} Let the value and weight  of item $i \in \cI, |\cI| = n$, be $v(i)$ and $w(i)$, respectively, and the corresponding weight to value ratio (called the buck-per-bang in the paper) be $b(i) = \frac{w(i)}{v(i)}$. The usual knapsack problem is to select the subset of items of $\cI$ that maximizes the sum of their values, subject to a hard constraint $C$ on the sum of the weight of the items in the selected set.
Without loss of generality, let $C=1$ by rescaling weights and $w(i) \le 1, \ \forall \ i$.

In this paper, we consider the knapsack problem with a slightly weaker constraint on capacity. Specifically, we assume that the capacity constraint is in expectation, i.e., an algorithm is allowed to violate the hard capacity constraint of $1$ on specific instances of input or its own randomization, but in expectation should meet the capacity constraint of $1$. This generalization is motivated by several practical cases of importance such as job scheduling in clouds, where typically the resource guarantees are easier to adhere to in expectation.

We consider the online version of the knapsack problem, where on each item's  arrival, it has to be accepted/rejected irrevocably. In the online setting, the performance metric is called the competitive ratio, that measures the ratio of the profit made by an online algorithm and the optimal offline algorithm that is allowed to know the future sequence (value and weight) of items, minimized over all possible input sequences $\sigma$, that specifies the order of arrival of items in $\cI$. Thus, for an algorithm $A$, its competitive ratio is 
$$ \mu_A = \min_{\sigma}\frac{\sum_{s\in S_A}v_{(\sigma)}(s)}{v_{(\sigma)}(\mathsf{OPT})},$$ where $\mathsf{OPT}$ is the optimal offline set of selected items and $S_A$ is the set of items selected by $A$. Hence the objective is to design an online algorithm with maximum competitive ratio.

For the case of hard capacity constraint, it is easy to see that if the input (values and weights of items) are chosen adversarially, no deterministic online algorithm can have bounded competitive ratio, and no randomized algorithm can have competitive ratio better than $1/n$, where $n$ is the total number of items. With capacity constraint in expectation, it still turns out that no randomized algorithm can have competitive ratio better than $1/n$. 
\begin{theorem} Under the expected capacity constraint of $1$, the competitive ratio of any online algorithm with the adversarial input is at most $1/n$.
\end{theorem}
\begin{proof} See Appendix \ref{app:adversarial}.\end{proof}

Following prior work, thus, to keep the problem non-degenerate in terms of competitive ratio, we assume that the order of arrival of items is uniformly random (secretary-model), i.e., each permutation over $n$ arriving items in $\cI$ is equally likely. 
Let $\pi$ be a uniformly random permutation over $[1:n]$. Then the the $k^{th}$ item that 
arrives has value $v(\pi^{-1}(k))$, weight $w(\pi^{-1}(k))$, and buck-per-bang $b(\pi^{-1}(k))$. 

For a set $S$, we let $v(S)  = \sum_{s \in S} v(s)$. Under the secretary-model of input, the competitive ratio of an online algorithm $A$ for solving the knapsack problem is defined as 
$$ \mu_A = \min_{\cI}\frac{\bbE_{\pi}\left\{\sum_{s\in S_A}v(s)\right\}}{v(\mathsf{OPT})},$$ where $\bbE$ is the expectation operator, $\cI$ is the complete set of items,  $\mathsf{OPT}$ is the optimal offline set of selected items and $S_A$ is the set of items selected by $A$. The online knapsack problem is to find the best algorithm $A$ that maximizes the competitive ratio $\mu_A$. $A$ is said to be $\alpha > 1$ competitive if $\mu_A =1/ \alpha$.
We first consider the two popular special cases of the knapsack problem, called the {\it secretary} and the $k$-secretary problems, where the weight of each item is $1$, before studying the general knapsack problem in Section \ref{sec:kp}.

\subsection{Secretary Problem}\label{sec:1sec}
In secretary problem, under the secretary-input model, the problem is to maximize the probability of selecting the best secretary (item with the largest value in our setting). Letting each item's weight to be $1$,  the classical secretary problem is a special case of the knapsack problem with hard capacity constraint of $1$, since  at most one item can be selected, and once the item is selected, the algorithm terminates. 

With the expected capacity constraint of $1$, as considered in this paper, the fundamental difference compared to the hard capacity constraint is that any online 
algorithm can actually access the whole input sequence sequentially and does not have to terminate as soon as one item is selected. However, an item is selected only using causal information, and once an item is selected, it cannot be removed subsequently.

Let $i^\star$ be the best item in $\cI$, then, under the secretarial input, the competitive ratio for algorithm $\cA$ for the secretary problem is defined as 
$$ \mu_A = \min_{\cI}\bbP_{\pi}(\cA \ \text{selecting item} \ i^\star),$$
with expected number of selected items being at most $1$. We next propose a simple modification to the classical solution to the secretary problem under the hard capacity constraint, and show that the competitive ratio can be improved significantly under the expected capacity constraint compared to the hard capacity constraint.

\subsubsection{Upper Bound on Competitive Ratio}
Consider a class of algorithms which we call $t$-Threshold Algorithm, that rejects the first $t$ items, and selects any item thereafter, if it is better than the best seen so far. Recall that for hard capacity constraint of $1$, optimal $t=n/e$ and the optimal algorithm terminates as soon as the first item that is better than the best items seen until $t=n/e$ is encountered. With the expected capacity constraint, the $t$-Threshold Algorithm does not terminate without considering the whole input sequence but once an item is selected, it cannot be rejected, and it has to choose $t$ judiciously to ensure the expected capacity constraint.

\begin{algorithm}
\caption{$t$-Threshold Algorithm}\label{alg:t}
\begin{algorithmic}[1]
\State \%{\bf Offline Phase} 
\State Do not select the first $t$ items 
$\cI_t(\pi) \subset \cI$ under permutation $\pi$  
\State $R$ = best item of $\cI_t(\pi)$
\State \%{\bf Decision/Online Phase} 
\State Initialize $S = \Phi$ \%The set to be selected
\State For every new item $i> t$ in the decision phase
\If{$v(i) > v(R)$}

\State Select item $i$
		\State  $S = S\cup\{i\}$
		\State  $R = \{i\}$
		
	\Else \State Do not select item $i$
\EndIf
\end{algorithmic}
\end{algorithm}
\begin{theorem}\label{thm:1sec} $t$-Threshold Algorithm with $t=n/e$\footnote{Throughout, for ease of exposition, we assume that $n/e$ is an integer, otherwise, a floor operator will be needed.}  has a competitive ratio of $1-1/e$ and it satisfies the expected capacity constraint of $1$. 
\end{theorem}
\begin{proof} For $t$-Threshold Algorithm with $t=n/e$, it is easy to see that the globally best item is not selected only if appears in the offline phase, i.e., it belongs to $\cI_t$, which happens with probability $1/e$. Thus, the probability of selecting the globally best item is $1-1/e$. 

Next, we check that the algorithm satisfies the expected capacity constraint.
Let $\b1_{\ell}$ be the indicator function that the item appearing at the $\ell^{th}$ location is selected by the algorithm. Then the number of selected items by the algorithm is $\sum_{\ell=n/e+1}^{n} \b1(\ell)$. By the definition of the algorithm, an item arriving at location $\ell$ is selected only if it is the best item seen so far, which happens with probability $1/\ell$. Thus,  $\b1_{\ell} = 1$ with probability $1/\ell$.
Using linearity of expectation, we have that the expected number of selected items is 
\begin{align}\nn
\bbE\{\# \text{selected items}\} &\le \sum_{\ell=n/e+1}^{n} \bbE\{\b1_{\ell}\} =  \sum_{\ell=n/e+1}^{n}1/\ell,\\\nn
&\le \int_{n/e}^n \frac{1}{x}dx,\\
&\le 1.
\end{align}
\end{proof}

In addition to having the expected number of selected items to be less than $1$, it is useful to know the distribution of the number of selected items by the $t$-Threshold Algorithm. For that purpose, in Fig. \ref{fig:hist} with $n=10000$ items, we plot the histogram of the number of items selected by the $t$-Threshold Algorithm with $t=\lfloor n/e\rfloor + 1$ to illustrate that not only the expected number of items selected is at most $1$, but there is rapid fall in the number of selected items and it almost never exceeds more than $6$ items.

\begin{figure}
\centering
\includegraphics[width=3.5in]{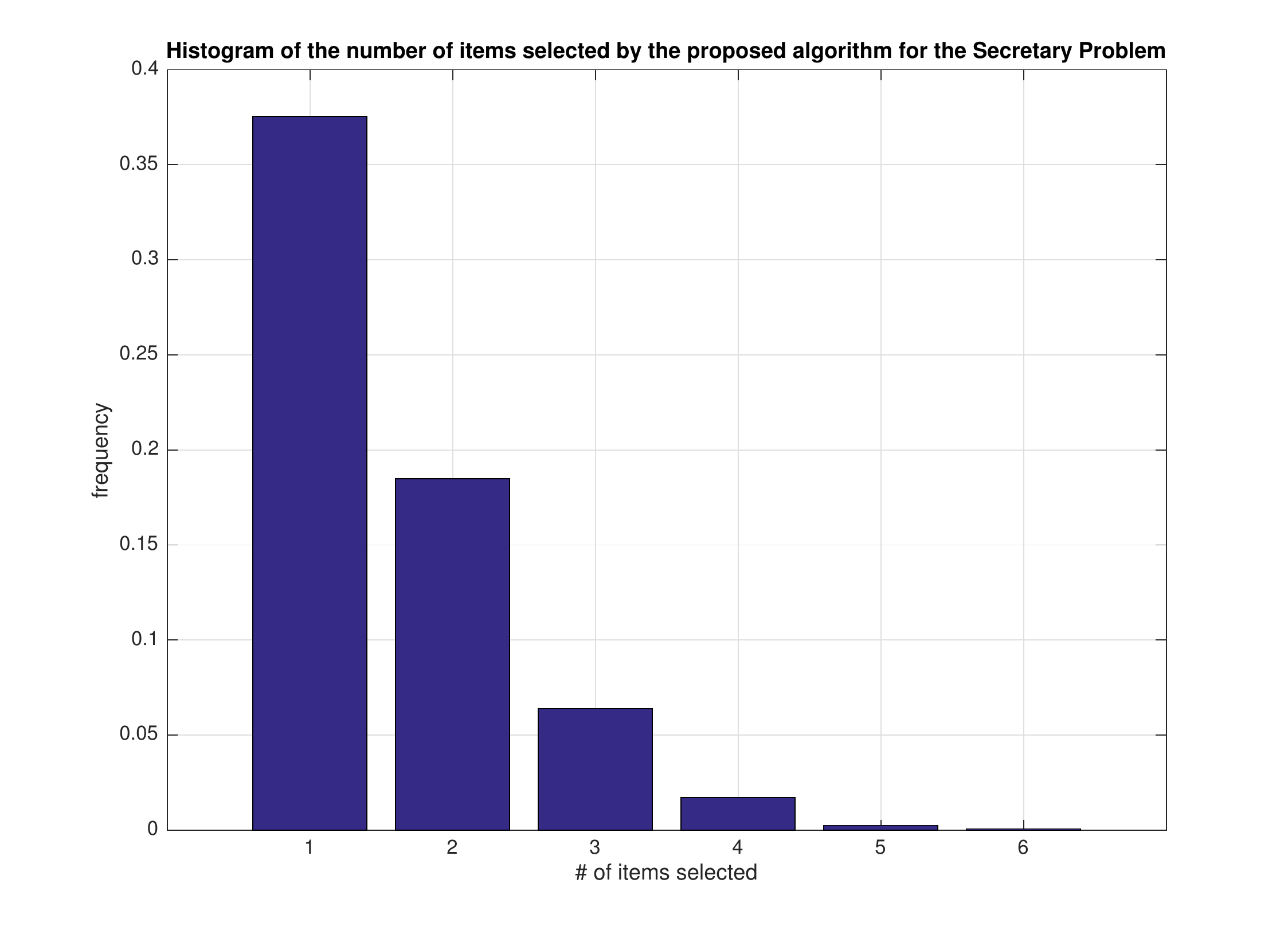}
\caption{Histogram of the number of items selected by the $t$-Threshold Algorithm for the Secretary Problem with number of items $n=10000$.}
\label{fig:hist}
\end{figure}

The $t$-Threshold Algorithm is a simple extension of the optimal algorithm to solve the secretary problem under the hard capacity constraint, where an item in the online phase that is better than the best seen in the offline phase is selected and the algorithm terminates. By choosing the length of the offline phase to be $n/e$, the classical result is that the best competitive ratio for the secretary problem under the hard capacity constraint of $1$ is $1/e$. What we show in Theorem \ref{thm:1sec} is that when the capacity constraint is in expectation, one can expect a two-fold increase in competitive ratio from $1/e$ to $1-1/e$, by selecting as many items that are better than the best seen so far starting from the item that arrives at location $n/e+1$. Thus, the relaxation in the capacity constraint allows a significant improvement in terms of selecting the best candidate.

We next show that no online algorithm can achieve better competitive ratio than $1-1/e$ under the expected capacity constraint of $1$.
\subsubsection{Lower Bound on Competitive Ratio}\label{sec:1seclb}
Now we try and argue that the competitive ratio of any online algorithm cannot be more than $1-1/e$ for solving the secretary problem under expected capacity constraint of $1$. Following observation is immediate, since we are trying to select only the best item and maximizing the probability of its selection. 
\begin{Observation}\label{obs:rec} An optimal algorithm will not select an item arriving at location $i$ if it is not the best seen so far. Moreover, if an optimal algorithm for solving the secretary problem under expected capacity constraint of $1$ selects an item arriving at the $i^{th}$ location, then it always selects any item that arrives after the $i^{th}$ location with the largest value so far.
\end{Observation}

\begin{theorem}\label{thm:1sec:lb} No online algorithm for solving the secretary problem under the expected capacity constraint of $1$ can have competitive ratio better than $1-1/e$. 
\end{theorem}
\begin{proof}
Consider an optimal algorithm $\opt$ for solving the secretary problem under expected capacity constraint of $1$. Let $p_r$ be the probability that it selects $r$ items at the end of the input sequence $\sigma$. Let $q_r$ be the probability that $r\ge 1$ items selected by $\opt$ contain the best item.

Thus, the lower bound on the competitive ratio under the expected capacity constraint of $1$ is 

\begin{equation}
\label{eq:LPsec}
\begin{array}{c l} \max &\sum_{r=1}^n p_r q_r, \\
\sum_{r=1}^n r p_r  \le 1,\\
p_r \in [0,1], \ \forall \ 1\le r\le n.
\end{array}
\end{equation}

Then in light of Observation \ref{obs:rec}, we have $q_r=1$ for $r>0$, since $\opt$ will select exactly $r>0$ items only when the globally best item is the $r^{th}$ item to be selected, otherwise no item is selected. Thus, we get the lower bound \eqref{eq:LPsec} as 
\begin{equation}
\label{eq:LPsec2}
\begin{array}{c l} \max &\sum_{r=1}^n p_r \\
\sum_{r=1}^n r p_r  \le 1,\\
p_r \in [0,1], \ \forall \ 1\le r\le n.
\end{array}
\end{equation}Thus, equivalently we have to solve for 
\begin{equation}
\label{eq:LPsec3}
\begin{array}{c l} \min &p_0, \\
\sum_{r=1}^n r p_r  \le 1,\\
p_r \in [0,1], \ \forall \ 1\le r\le n,
\end{array}
\end{equation}
where $\sum_{r=1}^n p_r= 1-p_0$, $p_0$ is the probability that $\opt$ misses out on selecting the globally best item. 
To minimize $p_0$, the optimal algorithm needs to start selecting items arriving at the earliest location possible, but in light of Observation $1$, if $\opt$ selects item at location $i$, then it will always select better items arriving after location $i$, increasing the number of selected items. An item arriving at location $i$ is selected by $\opt$ only if it is the best item seen so far, which happens with probability $1/i$. Thus, if $i$ is the first location at which $\opt$ decides to select an item if that item is the best seen so far, then $p_0 = 1-i/n$, and the expected number of items selected by the algorithm is $\sum_{j=i}^n \frac{1}{j}$.
Thus, \eqref{eq:LPsec3} is equivalent to 
\begin{equation}
\label{eq:LPsec3}
\begin{array}{c l} \min &1-i/n, \\
\sum_{j=i}^n \frac{1}{j}  \le 1, 1\le i\le n.
\end{array}
\end{equation}
Clearly, $\sum_{j=i}^n \frac{1}{j}$ can be well approximated by $\int_i^n \frac{1}{x} dx$. Hence, the constraint $\sum_{j=i}^n \frac{1}{j} \approx \int_i^n \frac{1}{x} dx \le 1$ implies that $i > n/e$, which implies that $p_0 \le 1-1/e$.
\end{proof}

Theorem \ref{thm:1sec}  and \ref{thm:1sec:lb} together characterize the optimal competitive ratio for the secretary problem under the expected capacity constraint. The classical secretary problem with hard capacity constraint is a richly studied object whose mutliple variants have been studied. To the best of our knowledge,  enforcing the capacity constraint in expectation is rather novel, and the more interesting upshot is to note that with a relaxation in capacity constraint, there is a significant improvement in the competitive ratio, and the optimal competitive ratio can be exactly characterized.

Next, we consider the natural generalization of the secretary problem, where more than one secretary can be selected, called the $k$-secretary problem.

%
%

\subsection{$k$-Secretary Problem}\label{sec:ksec}
In the classical $k$-secretary problem with a hard capacity constraint, each item has weight $1$ and an online algorithm can select at most $k$-items so as to maximize
\begin{equation}\label{eq:probksec} 
\mu_A =  \min_{S \subseteq \cI, |S|\le k} \frac{\bbE_{\pi}\{v(S)\}}{v(\cI^\star_k)}.
\end{equation}
where $S$ is the set of items selected by $A$ with $|S|\le k$,
$v(S) = \sum_{s\in S}v(s)$ for any subset $S \subseteq \cI$, and $\cI^\star_k$ is the best $k$-sized subset of $\cI$ in terms of the sum of the values.

With the expected capacity constraint of $k$, the objective function remains the same as in \eqref{eq:probksec}, except now the constraint is that the set of items $S$ selected by an online algorithm should satisfy $\bbE\{|S|\}\le k$.


Using linearity of expectation, to find a lower bound on the competitive ratio \eqref{eq:probksec}, it is sufficient to focus on minimum probability of selecting any item that belongs to the optimal subset $\cI^\star_k$. Towards that end, we propose a simple modification to the $t$-Threshold Algorithm as follows.

\begin{algorithm}
\caption{K-Sec $t$-Threshold Algorithm}\label{alg:tk}
\begin{algorithmic}[1]
\State \%{\bf Offline Phase} 
\State Do not select the first $t$ items 
$\cI_t(\pi) \subset \cI$ under permutation $\pi$  
\State $R$ = best $k$-item subset of $\cI_t(\pi)$
\State Order the items of $R$ in decreasing $\text{value}(r), r\in R$, the item with the least value is $r_{k}$
\State \%{\bf Decision/Online Phase} 
\State Initialize $S = \Phi$ \%The set to be selected
\
\State For every new item $i\ge t$ in the decision phase
\If{$v(i) > v(r_k)$}

\State Select item $i$
		\State  $S = S\cup\{i\}$
		\State  $R = R \cup \{i\} \backslash \{r_k\}$
		\State Order the items of $R$ in decreasing order of value
		\State $r_k$ is the item of $R$ with least value.
	\Else \State Do not select item $i$
\EndIf
\end{algorithmic}
\end{algorithm}
\begin{theorem}\label{thm:ksec} K-Sec $t$-Threshold Algorithm with $t=n/e$ is an optimal online algorithm for the $k$-secretary problem, with competitive ratio $1-1/e$  and  satisfies the expected capacity constraint of $k$. 
\end{theorem}
\begin{proof} With $t=n/e$, it is easy to see that any item belonging to the set $\cI_k^\star$ is not selected only if it appears in the offline phase $\cI_t$, which happens with probability $1/e$. So any item in $\cI_k^\star$ is selected with probability $1-1/e$. Therefore, the competitive ratio of K-Sec $t=n/e$-Threshold Algorithm, following \eqref{eq:probksec}, is $1-1/e$.

So we only need to check that if the algorithm satisfies the expected capacity constraint. 
Let $\b1_{\ell}$ be the indicator function that the item appearing at the $\ell^{th}$ location is selected by the algorithm. Then the number of selected items by the algorithm is $\sum_{\ell=n/e+1}^{n} \b1(\ell)$. By the definition of the algorithm, item arriving at location $\ell$ is selected only if it is among the $k$ best items seen so far, which happens with probability $k/\ell$. Thus,  $\b1_{\ell} = 1$ with probability $k/\ell$.
Using linearity of expectation, we have that the expected number of selected items is 
\begin{align}\nn
\bbE\{\# \text{selected items}\} &\le \sum_{\ell=n/e+1}^{n} \bbE\{\b1_{\ell}\} =  k\sum_{\ell=n/e+1}^{n}1/\ell,\\\nn
&\le k\int_{n/e}^n \frac{1}{x}dx,\\
&\le k.
\end{align}


The optimality of the algorithm follows from Theorem \ref{thm:1sec:lb}, since the competitive ratio is lower bounded by $1-1/e$ even for the $k=1$-secretary problem. 
\end{proof}
Thus, exploiting the linearity of expectation, we can get the same competitive ratio of $1-1/e$ for the $k$-secretary problem with $k>1$ similar to the $k=1$ case. This behaviour is identical to the case of hard capacity constraint, where also the optimal competitive ratio is $1/e$ for all values of $k$. Thus, relaxing the hard capacity constraint to an expected capacity constraint has identical performance advantage for the $k$-secretary problem independent of the value of $k$. Now, we are ready to consider the general online knapsack problem, where the weight of items is arbitrary, under the expected capacity constraint.

\section{Knapsack Problem}\label{sec:kp}
In this section, we consider the general knapsack problem under the expected capacity constraint of $1$. 
We will take a different approach for solving this general case compared to the special case of $k$-secretary problem studied in last subsection, where the weights of all items were identical. 

Before dealing with the online version of the knapsack problem, it is instructive to discuss its linear programming (LP) relaxation offline version, where each item can be selected fractionally, as follows.
The LP formulation for the fractional offline knapsack problem with knapsack size $C$ is given by,  
\begin{equation}
\label{eq:LP}
\begin{array}{c l} \max &\sum_{i \in \cI} v(i) x(i), \\
 \sum_{i \in \cI} w(i) x(i) \le C,\\
x(i) \in [0,1], \ \forall \ i\in \cI,
\end{array}
\end{equation}
where we have relaxed the condition that $x(i) \in \{0,1\}$ to $x(i) \in [0,1]$.

The following two facts are well-known \cite{babaioff2007knapsack} for the fractional knapsack problem.

\begin{lemma}\label{lem:fact1kp}
Recall that $\sfb(i) = \frac{w(i)}{v(i)}$.
There exists a threshold $\sfb^\star(C)$, such that all items $i$ with $\sfb(i) < \sfb^\star(C)$ are selected completely $x^\star(i)=1$ in \eqref{eq:LP}, while items $j$ with $\sfb_j > \sfb^\star(C)$ are not selected at all, $x^\star(j)=0$ in  \eqref{eq:LP}.
The only non-triviality is for items with $\sfb(i) = \sfb^\star(C)$, where the relaxed solution may be non-integral. 
\end{lemma}
To prove this, arrange the items in increasing order of $\sfb(i)$. Let there be an index $j$ such that $x(j) < 1$ but  $x(j+1) >0$. Then claim that $x(j) = x(j)+x(j+1) \frac{w(i+1)}{w(i)}$ and $x(j+1) = 0$ increases the value of the objective function, while still being capacity feasible.

Let $x_C^*(i)$ be the optimal fractional solution \eqref{eq:LP} with capacity $C$ and the corresponding optimal value of \eqref{eq:LP} be 
$v_{C}(\cI) = \sum_{i \in \cI} v(i) x_C^*(i)$.
\begin{lemma}\label{lem:fact2kp}
For $C_2\ge C_1$, we have that 
\begin{equation}\label{eq:scaling} 
v_{C_2}(\cI) \le \left(\frac{C_2}{C_1}\right) v_{C_1}(\cI).
\end{equation} \end{lemma}
Note that property \eqref{eq:scaling} may not be true for an integral optimal solution.

We will first define an offline knapsack algorithm that is allowed to use a larger capacity $C >1$, similar to \cite{iwama2010online}.  We then make it online with the help of sample and price class of strategies e.g. $t$-Threshold algorithm used in typical secretary or $k$-secretary problems, where the algorithm only observes (but does not select any) an initial set of items and builds a threshold, which is then used to select the forthcoming items.

{\bf Algorithm $\off$}: Order all the items in $\cI$ in non-decreasing order of their buck-per-bang $\sfb(i)$. Select as many items in the indexed order starting from the first, subject to the augmented capacity constraint of $C$. Thus, $\off$ selects the first $k$ indexed items if $\sum_{i=1}^k w(i) \le C$, and $\sum_{i=1}^{k+1} w(i) > C$. Let $\sfb^\star$ be the threshold on the buck-per-bang of all items selected by the $\off$ algorithm, i.e., 
$\sfb(i) \le \sfb^\star$ for all $i = 1,\dots, k$.

\begin{lemma}\label{lem:aug} Algorithm $\off$ with $1<C\le 2$ is $(C-1)$-approximate to the optimal solution of the offline knapsack problem \eqref{eq:LP} with hard capacity constraint $1$, where an offline algorithm is $\alpha$ approximate if the profit of the algorithm is at least $\alpha<1$ times the 
optimal offline algorithm's profit.
\end{lemma} 
\begin{proof} Consider the set $S_{\opt}$ and $S_{\off}$ of items selected by the optimal fractional offline algorithm with capacity $1$ \eqref{eq:LP}, and the algorithm $\off$ with capacity $C$ from the full set of items $\cI$, respectively. By definition, the value obtained by $S_{\opt}$ with capacity $1$ is $v_{1}(\cI) = \sum_{i \in \cI} v(i) x_1^*(i)$, and the value of $\off$ is $v_{S_\off}(C) = \sum_{i=1}^k v(i)$.


By definition, the set $S_{\off}$ is the set of $k$ items ordered in non-decreasing order of their buck-per-bang, where the first $k$ of them satisfy $\sum_{i=1}^k w(i) \le C$, and $\sum_{i=1}^{k+1} w(i) > C$ where items are indexed in non-decreasing order of their buck-per-bang. 

Moreover, since $w(i) < 1$ for each item $i$, we have 
\begin{equation}\label{eq:aug1}\sum_{i=1}^{k} w(i) > C-1.
\end{equation}

Therefore, \eqref{eq:aug1} implies that the first $k$ items of $\cI$ indexed in non-decreasing order of buck-per-bang require capacity more than $C-1$. Since the fractional optimal solution \eqref{eq:LP} also selects items in non-decreasing order of their buck-per-bang (Lemma \ref{lem:fact1kp}), we get that if the knapsack capacity was $C-1$, then the optimal fractional solution \eqref{eq:LP} would not have selected any item that is not selected by $\off$, i.e., $x^\star(k+1) = 0$ in \eqref{eq:LP} with capacity $C-1$.
Thus, we get  
\begin{equation}\label{eq:aug2}
v_{C-1}(\cI) \le v_{S_\off}(C) = \sum_{i=1}^k v(i).
\end{equation} 
Moreover, from Lemma \ref{lem:fact2kp}, with $C \le 2$, 
$$ v_1(\cI) \le \frac{v_{C-1}(\cI)}{C-1},$$ which combining with \eqref{eq:aug2}, we get 
$$ v_1(\cI) \le \frac{\sum_{i=1}^k v(i)}{C-1} = \frac{v_{S_\off}(C)}{C-1}, $$
proving the claim.

\end{proof}

\subsection{$2e$-Competitive Online Algorithm with Capacity $C$}
Before prescribing an online algorithm for the knapsack problem under the expected capacity constraint of $1$, we first take a detour via proposing an online version of the algorithm $\off$ that uses augmented capacity $C$. In the online setting, we aim to select as many items among the set of items $S_{\off}$ selected by $\off$ with capacity $C$ that have buck-per-bang greater than on equal to $\sfb^\star$.

To achieve this objective, following prior work \cite{VazeInfocom2017} we need to make an extra assumption (Assumption \ref{ass:randweight}) that is reasonable for most practical purposes. 
\begin{assumption}\label{ass:randweight} We assume that given two items arriving at locations $\pi(i)$ and $\pi(j)$, if $\sfb(i) > \sfb(j)$,  then $P(w(i)> w(j)) = \frac{1}{2}$ which is reasonable for most applications.
\end{assumption}

Consider the following online algorithm $\mathsf{AUG-ON}$. Divide the input into two phases; offline (first $n/e$ items) followed by online/decision (last $n(1-1/e)$ items). The algorithm observes the first $n/e$ items and does not select any of them. The offline algorithm $\off$ with capacity $C$ is run at the end of the offline phase (over the first $n/e$ items). Let $S_{\off}^{1/2}$ be the set of items that are selected by the $\off$ algorithm  in the offline phase, and the buck-per-bang threshold be $\sfb^\star_{1/2}$. 

In the online/decision phase, we will use a modified \textsc{Virtual} algorithm for the $k$-secretary problem \cite{babaioff2007knapsack}, starting from the arrival of $n/e+1^{st}$ item. At the beginning of the decision phase, we initialize the reference set as
$R = S_{\off}^{1/2}$, and $k= |S_{\off}^{1/2}|$. Thus, the algorithm aims to select  $k$ items as the 
$\off$ did in the sampling phase.
Moreover, in the decision phase only items with $\sfb(i) < \sfb^\star_{1/2}$ are eligible for selection, where the eventual selection is made if both the buck-per-bang and the weight of the newly arrived item is smaller than the buck-per-bang of the $k^{th}$ best item seen so far by the algorithm and if it was sampled in the offline phase. 

\begin{algorithm}
\caption{$\mathsf{AUG-ON}$ Algorithm}\label{alg:virtual}
\begin{algorithmic}[1]
\State \%{\bf Offline Phase} 
\State  Do not Select the first $t=n/e$ items 
$\cI_t(\pi) \subset \cI$ under permutation $\pi$,  
\State {Run $\off$ on subset of items $\cI_t(\pi)$ to get $S_{\off}^{1/2}$ and $\sfb^\star_{1/2}$} with capacity $C$
\State{Initialize $R= S_{\off}^{1/2}$, $k= |S_{\off}^{1/2}|$, $\sfb = \sfb^\star_{1/2}$}
\State $R = \{\text{$k$ \ largest item of } \cI_t(\pi)\} = \{i_1, \dots, i_k\}$ ordered in non-decreasing order of buck-per-bang values, such that $\sum_{i=1}^kw_i\le C$ and $\sum_{i=1}^{k+1}w_i > C$.  Item $i_k \in R$ has the largest buck-per-bang value.
\State \%{\bf Decision/Online Phase} 
\State Initialize $S = \Phi$ \%The set to be selected
\State For every new item $i$ in the decision phase with 
\If{$\sfb(i) < \sfb(i_k)$ }

\If  {$i_k$ was sampled in offline phase AND $w(i) \le  w(i_k)$ }
		\State  $S = S \cup \{i\}$
	\EndIf
	\State Update $R = R \backslash\{i_k\} \cup  \{i\}$
	
\State Order the item of $R$ in non-decreasing order of their buck-per-bang, item $i_k \in R$ has the largest value (worst item)
	\Else \State Do not select item $i$
\EndIf
\end{algorithmic}
\end{algorithm}
Algorithm $\mathsf{AUG-ON}$ is a modification of the $\mathsf{VIRTUAL}$ algorithm \cite{babaioff2007knapsack}, with the most important change being on line $10$ that is essential to ensure that the capacity constraint of $C$ is satisfied in the online/decision phase. We illustrate this first with an example as follows.

\begin{example} Let $S_{\off}^{1/2}$ the output of the offline phase contain four items with buck-per-bang 
$\{7,8,9,10\}$ with weight $w_1, w_2, w_3, w_4$, respectively, where $w_1+ w_2+ w_3+ w_4 \le C$. Let $R =S_{\off}^{1/2} = \{7,8,9,10\}$. Then on the arrival of a new item in the decision phase, with buck-per-bang $8.5$, 
it is included in $R$ by ejecting item $4$ with buck-per-bang $10$.  Thus, the updated reference set is $R = \{7,8,8.5,9 \}$ after rearranging the items in non-decreasing order of their buck-per-bangs. Moreover, the new item is selected as long as its weight is less than $w_4$. 
Thereafter, if an item with buck-per-bang $9.5$ arrives then it is neither selected nor included in $R$. Consider, one more item arriving with buck-per-bang $7.5$, then it is included in $R$ (selected only if its weight is less than $w_3$) and the updated set $R =  \{7,7.5, 8,8.5\}$. Hereafter, no more new items can be accepted since the item with the worst buck-per-bang $8.5$ is sampled in the decision phase. Important point to note in this example (that is a property of the algorithm) is that  an item in the decision phase is accepted only if its weight is less than a distinct item of $S_{\off}^{1/2}$, and once the weight of an item $i$ in $S_{\off}^{1/2}$ is compared with an item $j$ that arrives in the decision phase, then item $i$ is not available for future comparisons  irrespective of whether item $j$ was accepted or not. This leads us to following Lemma that the Algorithm $\mathsf{AUG-ON}$ satisfies the capacity constraint of $C$.
\end{example}
\begin{lemma}\label{lem:kpcap} Algorithm $\mathsf{AUG-ON}$ satisfies the capacity constraint of $C$, i.e., the sum of the weight of the items accepted in the decision phase is less than $C$.
\end{lemma}
\begin{proof}
Note that whether an item $i$ arriving in the decision phase is selected or not, as long as it is included in the set $R$, the item $j$ that is ejected from $R$ to make room for item $i$ which was sampled in the offline phase is never available thereafter for weight comparison for selection of new items. If the item with the worst buck-per-bang in $R$ is sampled in the decision phase, then no more items are selected anyway. Thus, importantly, the weight of any item that belonged to set $S_{\off}^{1/2}$ (output of $\off$) is compared at most once with any item arriving in decision phase. Thus, the weight of any item selected in the decision phase is less than the weight of any one distinct item of $S_{\off}^{1/2}$. Hence the sum of the weight of the items accepted in the decision phase is less than the sum of the weight of the items in $S_{\off}^{1/2}$, which is necessarily less than or equal to $C$ following the definition of $\off$.

\end{proof}

Let $i$ be an item selected by the algorithm $\off$ when run on the full set of items $\cI$. 
Then we will show that the probability of selecting item $i$ by this online algorithm is at least $\frac{1}{e}$. Thus, we have the following result.

\begin{lemma}\label{lem:augon}
The competitive ratio of the $\mathsf{AUG-ON}$ algorithm with respect to $\off$ is at least $1/e$. 
\end{lemma}
\begin{proof}
Let $\cI^\star_k$ be the set selected by the offline algorithm $\off$ when run on the full set of items $\cI$, with 
$\sfb^\star$ as the buck-per-bang threshold. Then the buck-per-bang threshold $\sfb^\star_{1/2}$ output by running $\off$ with capacity $C$ on $\cI_t(\pi)$ with $t=n/e$ satisfies $\sfb^\star_{1/2} \ge \sfb^\star$. Therefore, all items belonging to 
$\cI^\star_k$ are eligible for selection in the $\mathsf{AUG-ON}$ algorithm if they appear in the decision phase. 

For the moment, we disregard the weight acceptance condition that $w(i) \le w(i_k)$ on line $10$ of the algorithm. 

With the $\mathsf{AUG-ON}$ Algorithm, a new item that appears at location $s > t$ is selected if and only if at location $s$, the item with the largest buck-per-bang in the reference set $R$ is sampled at or before location $t$, and $\sfb(s) <\sfb(i_k)$. 

Since the permutations are uniformly random, the probability 
that at location $s$, the item with the smallest value in the reference set $R$ is sampled at or before time $t$ is $\frac{t}{s-1}$. Moreover,  the probability of any item $i\in \cI^\star_k$ arriving at the $s^{th}$ location is $\frac{1}{n}$ independent of $s$.

Hence the probability of selecting an item $i \in \cI^\star_k$ when it arrives at position $s \in [t+1, n]$, without considering the weight acceptance constraint that $w(i) \le w(i_k)$ on line $10$, is 
 \begin{eqnarray} \nn
P(i \in \cI^\star_k \ \text{is selected}) &=& \sum_{s=t+1}^n \frac{1}{n} \frac{t}{s-1} = \frac{t}{n} \sum_{s=t+1}^n\frac{1}{s-1} \\\nn
&>& \frac{t}{n} \int_t^n \frac{dx}{x}= \frac{t}{n} \ln\left(\frac{n}{t}\right).
\end{eqnarray}
Since we choose $t = \frac{n}{e}$, we get that $$P(i \in  \cI^\star_k \ \text{is selected}) = \frac{1}{e}.$$ 

Hence by linearity of expectation, we get that the expected value of the selected items $S$ by the  
\textsc{Virtual} algorithm is at least 
 \begin{equation}\label{eq:exppayoff}
 \bbE\left\{v(S)\right\} \ge \sum_{i \in  \cI^\star_k } \frac{1}{e} v(i)= \frac{1}{e}v(\cI^\star_k).
 \end{equation}
 
Now we enforce back the condition that $w(i) \le w(i_k)$ on line $10$ of the algorithm. As noted in Lemma \ref{lem:kpcap}, the weight of each of the $k$ items of $S_{\off}^{1/2}$ selected in the offline phase is compared at most once while selecting the new items in the online phase. 
Since for any two items $i,j$, given $\sfb(i) > \sfb(j)$, $P(w(i) > w(j)) =\frac{1}{2}$ from Assumption \ref{ass:randweight}, hence each item $i$ that is selected without enforcing $w(i) < w(r)$ for some item $r$ that is part of $S_\off^{1/2}$, is selected with probability $1/2$ even when the condition $w(i) \le w(r)$ is enforced, since each item $i$ selected by the \textsc{AUG-ON} algorithm has $\sfb(i) \le \sfb(j)$ for some distinct item $j$ that is part of offline selected set $S_\off^{1/2}$. 

Thus, we get from \eqref{eq:exppayoff}, that 
\begin{equation}\label{eq:exppayofffinal}
 \bbE\left\{v(S)\right\} \ge \frac{1}{2}\sum_{i \in  \cI^\star_k } \frac{1}{e} v(i)= \frac{1}{2e} v(\cI^\star_k),   
 \end{equation}
 and the competitive ratio of the \textsc{AUG-ON} algorithm is at least $\frac{1}{2e}$.

\end{proof}

Now we are ready to describe an online algorithm $\mathsf{ON}$ for the knapsack problem with the expected capacity constraint of $1$. We take recourse to the algorithm $\mathsf{AUG-ON}$ for that purpose.
\subsection{Online Algorithm with Expected Capacity Constraint}
\begin{algorithm}
\caption{$\mathsf{ON}$ Algorithm}\label{alg:on}
\begin{algorithmic}[1]
\State{Flip a fair coin}
\If{Heads}
\State{Run Algorithm $\mathsf{AUG-ON}$ with $C=2$}
\State{Accept all items accepted by $\mathsf{AUG-ON}$}
\Else{Tails}
\State{Do not select any item, Break;}
\EndIf

\end{algorithmic}
\end{algorithm}

\begin{theorem} The competitive ratio of algorithm $\mathsf{ON}$ is $1/4e$ and it satisfies the expected capacity constraint of $1$.
\end{theorem}
\begin{proof}
The online algorithm $\mathsf{ON}$ uses $\mathsf{AUG-ON}$ algorithm with $C=2$ with probability $1/2$, and does not choose any item with probability $1/2$. Therefore, clearly, the expected capacity constraint of $1$ is satisfied. From Lemma \ref{lem:aug}, it follows that $\off$ has approximation ratio of $\frac{1}{C-1}$ with respect to the optimal offline algorithm for the knapsack problem with hard capacity constraint of $1$. Moreover, Lemma \ref{lem:augon} ensures that $\mathsf{AUG-ON}$ has a competitive ratio of $1/2e$ with respect to the offline algorithm $\off$ run on full set of items $\cI$ with capacity $C$.
With the choice of $C=2$, $\mathsf{AUG-ON}$ is run with probability $1/2$, hence the overall competitive ratio of $\mathsf{ON}$ is $\frac{1}{4e}$ with respect to the optimal offline algorithm for the knapsack problem with hard capacity constraint of $1$.

\end{proof}
Thus with an expected capacity constraint, one can get far superior competitive ratio guarantees than the best known $1/10e$ guarantee \cite{babaioff2007knapsack} for the online knapsack problem under the hard capacity constraint. The basic idea of the proposed online algorithm is to use twice the capacity with probability $1/2$, so that the expected capacity constraint can be met, and to take advantage of the fact that with increased capacity, there is a simple threshold based algorithm ($\off$) that can closely approximate the optimal knapsack solution. The advantage of threshold based offline policy is that it can be made online with reasonable competitive ratio using the basic ideas developed for the $k$-secretary problem, where the objective is to choose each of the top $k$-item with large enough probability. One major challenge in the knapsack problem is that we do not know the exact number $k$ of items to be selected in contrast to the $k$-secretary problem. Our algorithm chooses the number of items to be selected as the number of items chosen by the threshold based offline algorithm $\off$ in the offline phase when run over a subset of items. Thus, $\off$ helps in finding a good threshold for selecting the items in the online phase, as well to find how many items to select.

\section{Conclusions}
In this paper, we have considered a new paradigm for some important online problems, namely the secretary and the knapsack problem, by relaxing the hard capacity constraint to an expected capacity constraint. This relaxation allows more flexibility for online algorithms that is well motivated by modern applications such as job scheduling in cloud servers, and is also an object of theoretical interest given the attention that both the secretary and the knapsack problem have received in literature. Under the expected capacity constraint we show that there is a two-fold increase in the competitive ratio for the $k$-secretary problem compared to the hard capacity constraint, which is significant. Moreover, for the knapsack problem, we are able to improve the competitive ratio by a factor of $2.5$ compared to the best online algorithm known under the hard capacity constraint. We believe that considering the expected capacity constraint is an exciting new direction that can be studied for online problems with hard capacity constraints that can allow fundamental improvement in the competitive ratios.

\appendices
\section{Competitive Ratio For Adversarial Input}\label{app:adversarial}
In this appendix, we show that under the adversarial input model, the competitive ratio of any online algorithm is at best $1/n$, even under the expected capacity constraint similar to hard capacity constraint.

Consider $n$ items with item $1$ being the best item. Let any online algorithm select $\ell \le n$ items with probability $p_\ell$. Index all the $\ell$-sized subsets $s_i, i=1,\dots,{n\choose \ell}$ of $[1:n]$ in lexicographic order $i$ coming before $j$ for $i<j$. Then an online algorithm selects $\ell \le n$ items arriving at locations defined by $s_i$ with probability $p_{\ell}(s_i)$ from the input sequence $\sigma$. Suppose $p_{\ell}(s_i) \ne \frac{1}{{n\choose k}} \ \forall \ s_i$, then the algorithm picks some subset locations with probability higher than
$\frac{1}{{n\choose k}}$; consequently some other subset locations will be picked with probability less than $\frac{1}{{n\choose k}}$. An adversary using this knowledge can put the best item to lie in any such subset locations, in which case the probability of selecting the best candidate will be less than $\ell/n$. Thus, with adversarial input, given that the algorithm is selecting $\ell$ items, the best strategy is to choose each of the location subsets equally likely.

Thus, the linear program to maximize the success probability for any online algorithm under the expected capacity constraint is
\begin{equation}
\label{eq:LPadv}
\begin{array}{c l} \max &\sum_{\ell=1}^n p_\ell \frac{\ell}{n}, \\
\sum_{\ell=1}^n \ell p_\ell  \le 1,\\
p_\ell \in [0,1], \ \forall \ 1\le \ell\le n,
\end{array}
\end{equation}
which is equivalent to 
\begin{equation}
\label{eq:LPadv2}
\begin{array}{c l} \max & \frac{1}{n}\sum_{\ell=1}^n {\tilde p_\ell} , \\
\sum_{\ell=1}^n {\tilde p_\ell}  \le 1,\\
{\tilde p_\ell} \in [0,1], \ \forall \ 1\le \ell\le n,
\end{array}
\end{equation}
where ${\tilde p} = \ell p.$ Thus, the maximum probability of success is at most $1/n$ with the adversarial input even when the capacity constraint is in expectation. Hence, given the expected capacity constraint of $1$, there is no advantage in choosing non-trivial probability distribution over the number of items to be selected when an adversary can choose the sequence of arrival.

\bibliographystyle{IEEEtran}
\bibliography{onlined2d}


\end{document}